\documentclass[a4paper, 10pt, conference]{ieeeconf}      

\IEEEoverridecommandlockouts                              

\usepackage{graphicx} 
\usepackage{epsfig} 
\usepackage{mathptmx} 
\usepackage{times} 
\usepackage{amsmath} 
\usepackage{amssymb}  
\usepackage{booktabs} 
\usepackage{cite}
\newtheorem{theorem}{Theorem}

\title{\LARGE \bf
Optimal first--passage time in gene regulatory networks
}
\author{Khem Raj Ghusinga$^{1}$ and Abhyudai Singh$^{2}$
\thanks{$^{1}$Khem Raj Ghusinga is with the Department of Electrical and Computer Engineering, University of Delaware, Newark, DE, USA 19716.
{\tt\small khem@udel.edu}}%
\thanks{$^{2}$Abhyudai Singh is with Faculty of Electrical and Computer Engineering, Biomedical Engineering, Mathematical Sciences, University of Delaware, Newark, DE, USA 19716.
        {\tt\small absingh@udel.edu}}%
}

\begin{document}
\maketitle
\thispagestyle{empty}
\pagestyle{empty}

\begin{abstract}
The inherent probabilistic nature of the biochemical reactions, and low copy number of species can lead to stochasticity in gene expression across identical cells. As a result, after induction of gene expression, the time at which a specific protein count is reached is stochastic as well. Therefore events taking place at a critical protein level will see stochasticity in their timing. First--passage time (FPT), the time at which a stochastic process hits a critical threshold, provides a framework to model such events. Here, we investigate stochasticity in FPT. Particularly, we consider events for which controlling stochasticity is advantageous. As a possible regulatory mechanism, we also investigate effect of auto--regulation, where the transcription rate of gene depends on protein count, on stochasticity of FPT. Specifically, we investigate for an optimal auto-regulation which minimizes stochasticity in FPT, given fixed mean FPT and threshold. 

For this purpose, we model the gene expression at a single cell level. We find analytic formulas for statistical moments of the FPT in terms of model parameters. Moreover, we examine the gene expression model with auto--regulation. Interestingly, our results show that the stochasticity in FPT, for a fixed mean, is minimized when the transcription rate is independent of protein count. Further, we discuss the results in context of lysis time of an \textit{E. coli} cell infected by a $\lambda$ phage virus. An optimal lysis time provides evolutionary advantage to the $\lambda$ phage, suggesting a possible regulation to minimize its stochasticity. Our results indicate that there is no auto--regulation of the protein responsible for lysis. Moreover, congruent to experimental evidences, our analysis predicts that the expression of the lysis protein should have a small burst size.
\end{abstract}

\section{Introduction}
Gene expression is the process of \textit{transcription} of genetic information to mRNAs, and \textit{translation} of each mRNA to proteins. As the copy number of species involved in the process is small, the probabilistic nature of biochemical reactions reflects as stochastcity in gene expression \cite{Blake_noise_2003, Raser_noiseocc_2005, ArjunRaj_nnc_2008, Munsky_noisetoregulation_2012, Kaern_TP_2005, AbhiMohammad_var_2013}.

\par Stochasticity in gene expression has an important role in several cellular functions. For example, it can lead genetically identical cells to different cell--fates \cite{Losick_cellfate_2008, Arkin_ska_1998, Weinberger_lentiviral_2005, Veening_bistability_2008, Hasty_switchamp_2000, abhi_transcriptionalburstingHIV_2010}. This helps the cells in responding to the ever--changing environment \cite{Eldar_functional_2010, Kussell_pd_2005, Balaban_bactpersist_2004, Murat_survival_2008}. On the other hand, stochasticity  in expression of housekeeping genes can lead to diseased states \cite{Kemkemer_incnoise_2002, Cook_modeling_1998, Bahar_incvariation_2006}, and needs to be minimized \cite{Lehner_noiseminm_2008, Fraser_noiseminm_2004}. Accordingly, different regulatory mechanisms are employed to control stochastic fluctuations \cite{UriAlon_nm_2007,Becskei_engineeringstability_2000, ElSamad_regulated_2006, Swain_attenuatestochasticity_2004, Orrell_control_2004, abhi_fbstrength_2009, Tao_effectoffb_2007, abhi_mRNA_2011}. Auto--regulation wherein transcription rate is a function of protein count is an example of one such mechanism. Its effect on stochascticity in gene expression has been a subject of several studies \cite{abhi_fbstrength_2009, Tao_effectoffb_2007, abhi_mRNA_2011}.

\par After onset of gene expression, its stochasticity consequently manifests into  stochasticity in the time at which a certain protein level is reached. This implies that the timing of a cellular event which triggers at a critical protein level is stochastic in nature \cite{Amir_noisetiming_2007, Murugan_fluctuation_2011}. For instance, lysis time for an \textit{E. coli} cell infected by a $\lambda$ phage virus is stochastic. Lysis of the cell takes place when holin, the protein responsible for lysis, reaches a critical threshold \cite{White_holintriggering_2011, JohnDennehy_lst_2011, Abhi_ltv_2014}.

\par Further, it has been suggested that optimality in lysis time provides evolutionary advantage to $\lambda$ phage virus \cite{INWang_fitness_2006, Wang_evolutiontiming_1996, Heineman_optimal_2007, Shao_adsorptionoptimal_2008, Bonachela_optimallysis_2014}. This indicates that there could be some regulation of gene expression to ensure lysis at the optimal time, with minimum stochastic fluctuations. In this work, we study stochasticity in first--passage time (FPT), the time it takes for the protein count to reach a fixed threshold for the first time \cite{redner2001guide}, at a single--cell level. We investigate the effect of auto-regulation of transcription on stochasticity of FPT. In particular we seek answer to the question: given the mean FPT (corresponding to optimal lysis time, for example), what auto-regulatory feedback will lead to minimum stochasticity in the FPT?

\par We first formulate an unregulated gene expression model assuming transcription, translation, and mRNA degradation while considering proteins to be stable. Along the lines of \cite{Abhi_ltv_2014}, we find expressions for statistical moments of FPT for this model, and discuss their implications with respect to minimizing variance in FPT for given mean FPT. Next, we introduce auto-regulation in the above model and derive the moments for FPT. Then, we deduce the expression for optimal feedback function that minimizes the variance in FPT for a given mean. We show that a negative or positive feedback always results into higher variance in first passage time for a given mean than the case when there is no feedback. The results are validated by carrying out simulations. Also, various notations used in this work are tabulated in Table I.
 
\begin{table}[h!t]
\caption{Description of notations used in this work}
\centering
\begin{tabular}{l p{6.5cm}}
\toprule
$k_m$					          & Transcription rate for unregulated gene expression model. \\
$k_p$                   & Translation rate for both unregulated, and regulated gene expression. models \\
$\gamma_m$				      & mRNA degradation rate for both unregulated, and regulated gene. expression models \\
$B_i$            		    & Burst size after $i^{th}$ transcriptional event. \\
$\mu$					          & Parameter of geometric distribution corresponding to. protein bursts \\
$b$					            & Mean of protein burst size. \\
$P(t)$ 		              & Protein count at time $t$. \\
$P_{i}$          		    & Protein count after $i^{th}$ burst. \\
$k_m(P_i)$       		    & Transcription rate for auto-regulated gene expression model after $i^{th}$ transcription event.\\
$X$            		      & Threshold for protein count. \\
$N$              		    & Minimum number of transcription events for protein count to reach the threshold $X$. \\
$T_i$            		    & Waiting time for $i^{th}$ transcription event. \\
$Y \sim \exp(\alpha)$		& $Y$ is an Exponential random variable with parameter $\alpha$. The probability density function of $Y$ is given by $f_Y(y)=\alpha e^{-\alpha y},\;y\geq 0$. \\
$f_N(n)$            	  & Probability mass function for minimum number of transcription events to reach the threshold $X$.\\
$f_{P_i}(j)$            & Probability mass function for protein count after $i$ transcription events. \\
$\left< .\right>$   	  & Expectation operator. \\
Var                     & Variance. \\
$k_{\max}$              & Maximum possible transcription rate in model with feedback implemented using Hill function .\\
$r$                     & Fraction of transcription rate $k_{\max}$ that corresponds to minimum transcription rate in model with feedback implemented using Hill function.\\
$H$                     & Hill coefficient. \\
$c$                     & Coefficient proportional to binding efficiency;                          decides when half rate concentration is reached. \\

\bottomrule
\end{tabular}
\label{tab:notations}
\end{table}
\section{First--Passage Time For Gene Expression Model Without Regulation}
In this section, we formulate a stochastic gene expression model (as shown in Fig. \ref{fig:geneexpression}). Then, we define the FPT for this model and derive expressions for its statistical moments. We also discuss the implications of these expressions in context of minimizing variance of FPT, for fixed mean and threshold. 
\subsection{Model Formulation}
\begin{figure}[h]
\centering
{\includegraphics[width=0.3\textwidth]{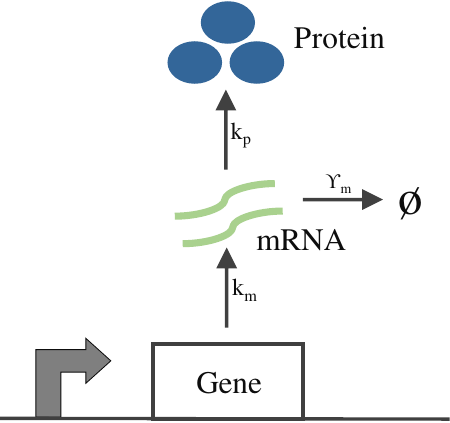}}
\caption{\textit{Model for gene expression without regulation}: The figure shows expression of a gene where mRNAs are transcribed from the gene at a rate $k_m$ and proteins are translated from each mRNA at a rate $k_p$. Proteins are assumed to be stable while each mRNA degrades with a rate $\gamma_m$. }
\label{fig:geneexpression}
\end{figure}
In the model under consideration transcription of mRNAs from the gene occurs at a rate $k_m$, translation of proteins from each mRNA occurs at a rate $k_p$, and each mRNA degrades at a rate $\gamma_m$. The time interval between two transcription events is exponentially distributed. We assume proteins to be stable as the lysis protein in $\lambda$ phage, i.e. holin, is stable \cite{Shao_holinstable_2009}. To further simplify the model, we assume each mRNA molecule degrades instantaneously after producing a burst of random number of protein molecules \cite{Friedman_pd_2006, ShahrezaeVahid_ad_2008, Paulsson_sge_2005, Berg_statfluct_1978}. Consistent with experimental, and theoretical evidences; we assume that protein burst follows a geometric distribution, and the mean burst size is given by $b=k_p/\gamma_m$ \cite{YuXiao_pom_2006, Elgart_connecting_2011}. Thus, the simplified model considers gene expression wherein each burst event (equivalent to transcription event) occurs at an exponentially distributed time with parameter $k_m$, and size of burst follows a geometric distribution with mean $b$.  

\par Let us denote the size of $i^{th}$ burst by random variable $B_i$ and the parameter of its distribution by $\mu$. The probability mass function, therefore, can be written as \cite{Degroot_prob_2012}:
\begin{equation}
\label{eqn:goemetricburstpmf}
\text{Pr}(B_i=k)=\mu \left(1-\mu \right)^k,\,\mu  \in (0,1],\,k \in \{0,1,2..\}.
\end{equation}
\par The mean burst size, $b$, can be expressed as \cite{Degroot_prob_2012}:
\begin{equation}
\left<B_i\right>=b = \frac{1-\mu}{\mu}.
\end{equation}
Further, let protein count after $n$ transcription events be denoted as $P_n$. It can be expressed as a sum of random variables $B_i$:
\begin{equation}
\label{eqn:proteincount}
P_n=\sum_{i=1}^{n} B_i.
\end{equation}
Being sum of independent and identically distributed geometric random variables, $P_n$ has a negative binomial distribution with parameters $n$ and $\mu$ \cite{Papoulis_prob_2002}. The probability mass function of $P_n$, denoted as $f_{P_n}(j)$, can be expressed as \cite{Papoulis_prob_2002}:
\begin{equation}
\label{eqn:proteincountpmf}
f_{P_n}(j)=\text{Pr}\left(\sum_{i=1}^{n}B_i = j \right) = {n+j-1 \choose n-1}\mu^n\left(1-\mu\right)^j.
\end{equation}
\par Also, the cumulative distribution function is given by \cite{Spiegel_scham_1992}:
\begin{equation}
\label{eqn:proteincountcdf}
\text{Pr} \left(\sum_{i=1}^{n}B_i \leq j \right) = 1- I_{1-\mu}(j+1,n),
\end{equation}
where $I_{1-\mu}(j+1,n)$ is regularized incomplete beta function: 
\begin{equation}
I_{1-\mu}(j+1,n) = \sum_{l=j+1}^{n+j} {n+j \choose l} (1-\mu)^l \mu^{j+n-l}, 
\end{equation}
and satisfies the following property: 
\begin{equation}
\label{eqn:Iproperty}
I_{1-\mu}(j+1,n) = 1-I_{\mu}(n,j+1).
\end{equation}
We have determined the distribution for protein population. Next, we defined the first--passage time (FPT) for the protein count to reach a certain threshold. 
\subsection{Expression for First Passage Time}
For a random process corresponding to protein count, $P(t)$, with $P(0)=0$, the first passage time (FPT), for a threshold $X$ is defined as:
\begin{equation}
\label{eqn:firstpassagetime}
FPT:= \inf \{t: P(t) \geq X \}, \quad X \in \{1, 2, 3, ...\}.
\end{equation}
Because in our model, the protein count changes only when a burst occurs (or equivalently, a transcription event occurs); we can calculate the minimum number of transcription events, $N$, it takes for the protein count to reach the threshold $X$ and define the FPT as sum of inter--burst arrival times. This has been depicted in Fig. \ref{fig:proteincount}. 
\begin{figure}[h!]
\centering
{\includegraphics[width=\columnwidth]{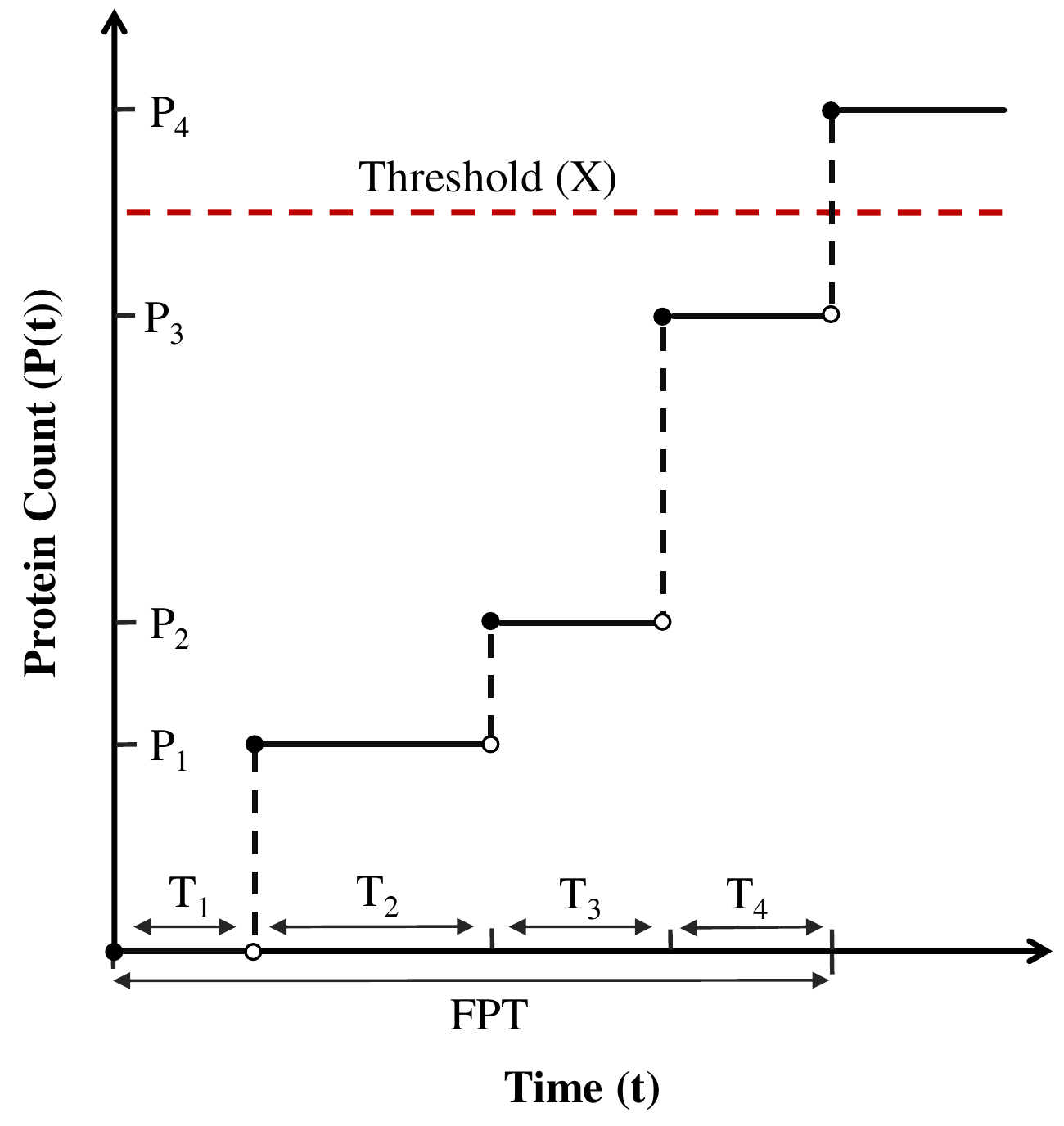}}
\caption{\textit{First--passage time for gene expression in burst limit}: The gene expresses in bursts which arrive at time intervals $T_i, \; i=1,2,...$. The protein count after $i^{th}$ burst is denoted by $P_i$. Protein count at time $t$ is denoted by $P(t)$, and is equal to $P_i$, where $i$ is number of bursts until time $t$. The first--passage time can be expressed as the sum of inter--burst arrival times till $P_i$ crosses the threshold $X$ for the first time.}
\label{fig:proteincount}
\end{figure}
\par Let the time between ${i-1}^{th}$ and $i^{th}$ bursts be denoted by random variable $T_i$, then:
\begin{equation}
\label{eqn:fptdef}
FPT=\sum_{i=1}^{N}T_i,
\end{equation}
\text{where $N$ is given by the following equation:}
\begin{equation}
\label{eqn:Ndefinition}
  N=\inf \left( n: P_n \geq X\right), \quad n \in \{1,2,...\}, \;\; X \geq 1.
\end{equation}
Note that in Eq. \eqref{eqn:fptdef}, $T_i$ are independent, and identically distributed exponential random variables with parameter $k_m$. We denote this by $T_i \sim \exp(k_m)$. Also, each of $T_i$ is independent of $N$. 
\par Using standard results from probability theory, one may write \cite{Ross_prob_2010}:
\begin{subequations}
\begin{align}
\left< FPT \right> &= \left<N \right> \left<T_i \right> \label{eqn:FPTMeanNOFB}, \\
\text{Var}(FPT) & = \left<N\right>\text{ Var }(T_i) + \text{ Var } (N) \left<T_i\right>^2 \label{eqn:FPTVarNOFB}.
\end{align}
\end{subequations}
It can be noted that to determine statistical moments of FPT in Eq. \eqref{eqn:FPTMeanNOFB}--\eqref{eqn:FPTVarNOFB}, we need to derive expressions for first two moments of $T_i$, and $N$.
\subsubsection{First Two Moments of $N$}
The cumulative distribution function for $N$ defined in Eq. \eqref{eqn:Ndefinition} can be written as:
\begin{subequations}
\begin{align}
\text{Pr} (N \leq n) &= \text{Pr} \left(P_n \geq X \right),  \\
                                       &= 1-  \text{Pr}\left(P_n \leq X-1 \right).
\end{align}
\end{subequations}

Since $\displaystyle P_n$ is a negative binomial distribution, we have:
\begin{subequations}
\begin{align}
\text{Pr} (N \leq n)  &= 1-\left(1-I_{1-\mu}(X,n)\right), \\
                            &= I_{1-\mu}(X,n) \label{eqn:cdfN1}.
\end{align}
\end{subequations}
Using the property of incomplete beta function mentioned in Eq. \eqref{eqn:Iproperty}, we get:
\begin{equation}																		
          \text{Pr} (N \leq n)=1-I_{\mu}(n,X) \label{eqn:cdfN}.
\end{equation}

Comparing with Eq. \eqref{eqn:proteincountpmf} and Eq. \eqref{eqn:proteincountcdf}, the probability mass function corresponding to Eq. \eqref{eqn:cdfN} can be written as:
\small{
\begin{equation}
\label{eqn:Npmf}
f_N(n)={n+X-2 \choose n-1}(1-\mu)^{n-1} \mu^{X},\;\;n \in \{1,2,...\}, \;X \geq 1.
\end{equation}
}
\normalsize
\begin{subequations}
First two statistical moments of the distribution in Eq. \eqref{eqn:Npmf} are given by \cite{Papoulis_prob_2002}:
\begin{align}
\left<N\right> &= \frac{\mu X}{1-\mu}+1=\frac{X}{b}+1, \label{eqn:meanN}\\
\text{Var}(N)&=\left<N^2\right>-\left<N\right>^2=\frac{\mu X}{(1-\mu)^2}=\frac{X}{b}\frac{1+b}{b}. \label{eqn:varN}
\end{align}
\end{subequations}
\subsubsection{First Two Moments of $T_i$}
Since $T_i \sim \exp(k_m)$, its statistical moments are given by:
\begin{subequations}
\begin{align}
\left< T_i \right> &= \frac{1}{k_m}  \label{eqn:TimeanIID},\\
\text{Var}(T_i) &= \left< T_i^2 \right>-\left< T_i \right>^2 = \frac{1}{k_m^2}=\left< T_i \right>^2\label{eqn:TivarIID}.
\end{align}
\end{subequations}
 
We now have expressions for first two moments of $T_i$, and $N$. The expressions for first two moments of FPT in terms of model parameters can, therefore, be written as:
\begin{align}
\left<FPT\right> &= \left(\frac{X}{b}+1\right)\frac{1}{k_m} \approx \frac{X}{bk_m}, \label{eqn:FPTmeanmodel}\\
\text{Var}(FPT) &= \frac{X(2b+1)+b^2}{b^2k_m^2}\approx \frac{X}{b^2k_m^2}(1+2b), \label{eqn:FPTvarmodel}
\end{align}
where the approximations are valid when $X \gg b$. It can be observed a smaller mean burst size $b$ would result in smaller variance of FPT. The mean FPT can be kept fixed by a commensurate change in the transcription rate, $k_m$. Therefore, the variance can independently be reduced by a lower mean burst size $b=k_p/\gamma_m$. This means adopting a high transcription rate $k_m$, and a low translation rate $k_p$ (and/or having a higher degradation rate $\gamma_m$ for the mRNAs) results in a lower variance in FPT without affecting its mean.

Further, we note that by using $\text{Var}(T_i)=\left<T_i\right>^2$ from Eq. \eqref{eqn:TivarIID}, we can deduce the following relationship between $\left<FTP\right>$ and $\left<FPT^2\right>$ from Eq. \eqref{eqn:FPTMeanNOFB} and Eq. \eqref{eqn:FPTVarNOFB}: 
\begin{equation}
\label{eqn:RelnFPTMomentsNoFB}
\left<FPT^2\right> = \frac{\left<FPT\right>^2}{\left<N\right>^2}\left<N^2\right>+\frac{\left<FPT\right>^2}{\left<N\right>}.
\end{equation}
We shall use above relationship in the later part of the paper while deriving expression of the auto-regulation function that minimizes variance in FPT, for given mean FPT. 
\par Next, we introduce auto-regulation of transcription rate by the protein count to investigate how the expressions for statistical moments of FPT change. 
\section{Introducing Auto-regulation in Gene Expression Model}
To investigate the effect of auto-regulation on statistical moments of FPT, we assume that transcription rate is a function of protein count, i.e., it changes after each transcription event.  We denote the transcription rate after arrival of $i^{th}$ burst as $k_m(P_{i})$.  Similar to previous section, we need to derive expression for moments of inter--burst arrival times $T_i$, and minimum number of transcription events $N$ in order to derive the expression for FPT moments defined in Eq. \eqref{eqn:fptdef}. 
\par We note that the translation burst size is independent of the transcription rate. Therefore, distribution of $N$ to reach a certain threshold $X$ is same as gene expression model without any regulation discussed in previous section. However, distribution of each $T_i$ is different and depends upon corresponding rate of transcription. 
\par We derive expressions for first two moments of each $T_i$ to find analytical forms of first two moments of FPT. 
\subsection{Inter--burst arrival time for auto-regulatory gene expression model}
It may be noted that if protein count after any burst event is known, arrival time for the next burst will be exponentially distributed. Therefore, the distribution of each $T_i$ can be modelled as a conditional exponential distribution. More specifically, we can write:
\begin{equation}
T_i \sim \exp \left(k_m(P_{i-1})| P_{i-1}\right),
\end{equation}

where $T_i$, and $P_{i-1}$ respectively denote the arrival time for $i^{th}$ burst and protein count after the ${i-1}^{th}$ burst. 

\par The expressions for mean and variance of $T_i$ can be calculated as follows.
\subsubsection{Mean}
Before arrival of the first burst, there are no protein molecules, i.e., $P_{i-1}=0$ for $i=1$. Therefore, we can write the mean for arrival time for the first burst as:
\begin{equation}
\label{eqn:meanT1}
\left<T_1\right>=\frac{1}{k_m(0)}.
\end{equation}
For  $i \in \{2, 3, 4...\}$, the corresponding arrival times would be conditionally exponential, implying: 
\begin{subequations}
\begin{align}
\left<T_i|P_{i-1}=j\right> &=\frac{1}{k_m(j)}, \\
\implies \left<T_i\right>&=\sum_{j=0}^{\infty} \frac{1}{k_m(j)} \text{Pr}\left(P_{i-1}=j\right), \\
                         &=\sum_{j=0}^{\infty}   \frac{1}{k_m(j)} f_{P_{i-1}}(j). \label{eqn:meanTi}
\end{align}
\end{subequations}
\subsubsection{Second Order Moments}
Adopting similar approach as above, we derive the expressions for second order moments of $T_i$. For $i=1$, we have:
\begin{subequations}
\begin{equation}
\label{eqn:secondmomentT1}
\left<T_1^2\right> = \frac{2}{k_m^2(0)}.
\end{equation}
For $ i \in \{2,3,4...\}$:
\begin{align}
\left<T_i^2|P_{i-1}=j\right> &=\frac{2}{k_m^2(j)}, \\
\implies \left<T_i^2\right>&=\sum_{j=0}^{\infty} \frac{2}{k_m^2(j)} \text{Pr}\left(P_{i-1}=j\right), \\
                         &=\sum_{j=0}^{\infty}   \frac{2}{k_m^2(j)} f_{P_{i-1}}(j). \label{eqn:secondmomentTi}
\end{align}
\end{subequations}
Therefore the expression for variance of $T_1$:
\begin{equation}
\label{eqn:varT1}
\text{Var} (T_1)=\frac{1}{k_m^2(0)}=\left<T_1\right>^2.
\end{equation}
For $i \in \{2,3,4,...\}$, the expression for $\text{Var}(T_i)$ will be
\begin{equation}
\label{eqn:varTi}
\text{Var}(T_i)=\sum_{j=0}^{\infty} \frac{2}{k_m^2(j)} f_{P_{i-1}}(j) - \left[ \sum_{j=0}^{\infty}  \frac{1}{k_m(j)} f_{P_{i-1}}(j) \right]^2.
\end{equation}

Moreover, we have following relationship first two moments of the random variable $1/k_m(P_{i-1})$:
\begin{equation}
\sum_{j=0}^{\infty} \frac{1}{k_m^2(j)} f_{P_{i-1}}(j) \geq \left[ \sum_{j=0}^{\infty}  \frac{1}{k_m(j)} f_{P_{i-1}}(j) \right]^2,
\end{equation}
which alongwith Eq. \eqref{eqn:varTi}, and \eqref{eqn:varT1} yields: 
\begin{equation}
\label{eqn:varmeanrelnTi}
\text{Var}(T_i)\geq \left< T_i\right>^2.
\end{equation}
We note that the equality above holds for $i=1$. We will use it in later part of the paper while deducing the expression for optimal auto-regulation that leads to minimum variance in the FPT for fixed mean. 
\par Having derived the expressions for moments of inter--bursts arrival times, we see how the introduction of auto-regulation influences the expressions for FPT moments. 
\subsection{FPT for auto-regulatory gene expression model}
We present the expressions for statistical moments of FPT in theorem--proof format. In developing the proofs, we make use of the fact that each $T_i$ will be independent of $N$. Also, $T_i$ are independent of each other. However, they are not identically distributed like the unregulated gene expression case discussed in previous section. 
\begin{theorem}[Mean of First Passage Time]\label{thm:meanFPT}
For the FPT defined in Eq. \eqref{eqn:fptdef}, the mean FPT is given by following expression: 
\begin{equation}
\left<FPT\right>=\sum_{n=1}^{\infty} \sum_{i=1}^{n}\left<T_i\right> f_N(n) \label{eqn:FPTMeanI},
\end{equation}
where $\displaystyle f_N(n)$ is defined in Eq. \eqref{eqn:Npmf}, $\displaystyle \left< T_i \right>$ is given by Eq. \eqref{eqn:meanT1}, \eqref{eqn:meanTi} and $\displaystyle \left< N \right>$ is given by Eq. \eqref{eqn:meanN}.
\end{theorem}
\begin{proof}
To prove the result, we first find conditional expectation given $N=n$ then we have: 
\begin{subequations}
\begin{align}
\left<FPT|N=n\right>&=\left<\sum_{i=1}^{n}T_i\right>, \\
                    &=\sum_{i=1}^{n}\left<T_i\right>.
\end{align}
Unconditioning above expression with respect to $N$:
\begin{align}
\left<FPT\right> &= \sum_{n=1}^{\infty} \sum_{i=1}^{n}\left<T_i\right> Pr(N=n), \\
                          &= \sum_{n=1}^{\infty} \sum_{i=1}^{n}\left<T_i\right> f_N(n). 
\end{align}
\end{subequations}
This completes the proof. 
\end{proof}
\begin{theorem}[Variance of First Passage Time]
For the FPT defined in Eq. \eqref{eqn:fptdef}, the variance of FPT is given by the following expression: 
\small{
\begin{equation}
\begin{aligned}
\sum_{n=1}^{\infty}\left(\sum_{i=1}^{n} \text{Var}\left(T_i\right) + \left(\sum_{i=1}^{n}\left<T_i\right>\right)^2 \right)f_N(n) - \left( \sum_{n=1}^{\infty} \sum_{i=1}^{n}\left<T_i\right> f_N(n)\right)^2,
\end{aligned}
\label{eqn:varFPT}
\end{equation}
}
\normalsize
where $\displaystyle f_N(n)$ is defined in Eq. \eqref{eqn:Npmf}, $\displaystyle \left< N \right>$ is given by Eq. \eqref{eqn:meanN}, $\displaystyle \left< N^2 \right>$ can be deduced from Eq. \eqref{eqn:varN}, $\displaystyle \left< T_i \right>$ is given by Eq.  \eqref{eqn:meanT1}, \eqref{eqn:meanTi} and $\text{Var}(T_i)$ is given by Eq. \eqref{eqn:varT1}, \eqref{eqn:varTi}.
\end{theorem}
\begin{proof}
 Since expression for $\left<FPT\right>$ is known and given by Eq. \eqref{eqn:FPTMeanI}, we need to find expression for $\left<FPT^2\right>$, in order to find expression for variance of FPT. 
 
\par Using the definition of first passage time in Eq. \eqref{eqn:fptdef}, we have:
\begin{subequations}
\begin{align}
      \left<FPT^2|N= n\right> &= \left<\sum_{i=1}^{n}\sum_{j=1}^{n}T_i T_j\right>, \\
                              &=  \left<\sum_{i=1}^{n}T_i^2+\sum_{i=1}^{n}\sum_{j=1 \neq i}^{n}T_i T_j\right>
\end{align}

Since $T_i^2$ are independent of each other, and  $T_j$ are independent of $T_i$ for each $j\neq i$; we can write:
\begin{align}
 \left<FPT^2|N= n\right> &= \sum_{i=1}^{n}\left<T_i^2\right> + \sum_{i=1}^{n}\sum_{j=1\neq i}^{n}\left<T_iT_j\right>, \\
                         &= \sum_{i=1}^{n}\left<T_i^2\right> + \sum_{i=1}^{n}\sum_{j=1\neq i}^{n}\left<T_i\right>\left<T_j\right>.
\end{align}
Using $\text{Var}(T_i)=\left<T_i^2\right>-\left<T_i\right>^2$, we have: 		
\begin{equation}								
\left<FPT^2|N= n\right> = \sum_{i=1}^{n} \text{Var}\left(T_i\right) + \left(\sum_{i=1}^{n}\left<T_i\right>\right)^2.
\end{equation}
\end{subequations}
Unconditioning with respect to $N$, expression for $\left<FPT^2\right>$ becomes:
\begin{equation}
\left<FPT^2\right> = \sum_{n=1}^{\infty}\left(\sum_{i=1}^{n} \text{Var}\left(T_i\right) + \left(\sum_{i=1}^{n}\left<T_i\right>\right)^2 \right)f_N(n).
\label{eqn:expectFPTsquared}
\end{equation}
Therefore, using Eq. \eqref{eqn:FPTMeanI}, and Eq. \eqref{eqn:expectFPTsquared}; expression for $\text{Var}(FPT)$ becomes: 
\small{
\begin{equation}
\begin{aligned}
\sum_{n=1}^{\infty}\left(\sum_{i=1}^{n} \text{Var}\left(T_i\right) + \left(\sum_{i=1}^{n}\left<T_i\right>\right)^2 \right)f_N(n) - \left( \sum_{n=1}^{\infty} \sum_{i=1}^{n}\left<T_i\right> f_N(n)\right)^2.
\end{aligned}
\end{equation}
}
This completes the proof.
\end{proof}
So far we have developed analytical expressions for mean and variance of FPT when there is an auto-regulatory feedback to transcription rate from protein count. In the next section, we  make use of these expressions to deduce the optimal auto-regulation function to minimize the variance of FPT assuming fixed mean FPT.

\section{Minimizing Variance in First Passage Time for Given Mean}

\begin{table*}[h!t]
\caption{Model parameters used for simulation of positive, negative, and no feedback cases. }
\centering
 \begin{tabular}{lllll}
\toprule 
  Parameter & Unit & Positive feedback & Negative feedback & No feedback   \\
  \midrule
  $k_{\max}$ & mRNA produced per minute &19.35 & 84 & 10 \\
  $k_p$ & protein produced per mRNA per minute & 2.65 & 2.65 & 2.65 \\
  $\gamma_m$ & per minute & 0.3 & 0.3 & 0.3 \\
  $X$ & molecules & 5000 & 5000 & 5000 \\
	$r$ & - &0.05 & 0.05 & - \\
  $c$ & per molecule &0.002 & 0.002 & - \\
	$H$ & - & 2 & 2 & - \\
\bottomrule
 \end{tabular}
\label{tab:parameters}
\end{table*}
In this section, we find expression for the auto-regulatory feedback function, $k_m(P_{i-1}),\;\; i \in \{1,2,3,...\}$ that gives minimum variance in FPT, given the mean FPT and event threshold are fixed. The result is presented in form of a theorem. 
\begin{theorem}[Optimal feedback for minimum variance]
Let the first passage time be defined as Eq. \eqref{eqn:fptdef}, and its mean and variance, respectively, given by Eq. \eqref{eqn:FPTMeanI} and Eq. \eqref{eqn:varFPT}. Then, the optimal function to minimize the variance of FPT for a given mean of FPT will be constant, given by following expression: 
\begin{equation}
\label{eqn:optfunc}
k_m\left(P_{i-1}\right)=\frac{\left<N\right>}{\left<FPT\right>}, \quad \forall i \in \{1, 2, 3, ...\},
\end{equation}
where $\left<N\right>$ denotes the minimum number of transcription events required to reach the FPT threshold, and is given by Eq. \eqref{eqn:meanN}. 
\end{theorem}
\begin{proof}
We assume that each burst event adds a perturbation to transcription rate, i.e., $1/k_m(P_{i-1})$ can be written as: 
\begin{equation}
\label{eqn:assumedoptimalfunction}
\frac{1}{k_m(P_{i-1})}:=\frac{\left< FPT \right>}{\left<N\right>}+\delta_i, 
\end{equation}
where $\delta_i$ is perturbation corresponding to transcription rate after $i-1^{th}$ burst. To prove the result, we shall prove that the variance of FPT for given mean will minimize when $\delta_i=0$.
  
\par Recalling the expression for $\left<FPT\right>$ from Eq. \eqref{eqn:FPTMeanI}:
\begin{equation}
\left<FPT\right> = \sum_{n=1}^{\infty} \sum_{i=1}^{n}\left<T_i\right> f_N(n).
\end{equation}

Using expressions in Eqs. \eqref{eqn:meanT1}, \eqref{eqn:meanTi}, we can deduce the expressions for $\left<T_i\right>$ as:
\begin{equation}
\label{eqn:optimalTiform}
\left<T_i\right>=\frac{\left< FPT \right>}{\left<N\right>}+\epsilon_i, 
\end{equation}
where $\epsilon_i$ is related with $\delta_i$ by following expression:
\begin{equation}
\label{eqn:epsdelrelation}
\epsilon_i:=\sum_{j=1}^{\infty}\delta_i f_{P_{i-1}}(j)=\left<\delta_i\right>.
\end{equation} 
Substituting expression for $\left<T_i\right>$ from Eq. \eqref{eqn:optimalTiform}, we have:
\begin{subequations}
\begin{align}
 \left<FPT\right> &= \sum_{n=1}^{\infty} \sum_{i=1}^{n} \left(\frac{\left<FPT\right>}{\left<N\right>}+\epsilon_i\right) f_N(n),  \\
   &= \sum_{n=1}^{\infty} \left(\frac{\left<FPT\right>}{\left<N\right>}n +  \sum_{i=1}^{n} \epsilon_i \right)f_N(n),  \\
   & = \frac{\left<FPT\right>}{\left<N\right>}\sum_{n=1}^{\infty}nf_N(n)  + \sum_{n=1}^{\infty} \sum_{i=1}^{n}\epsilon_i f_N(n).
\end{align}
\end{subequations}
Since $\displaystyle \sum_{n=1}^{\infty}nf_N(n)=\left<N\right>$, we have:
\begin{equation}
\label{eqn:epsiloncondition}
 \sum_{n=1}^{\infty} \sum_{i=1}^{n}\epsilon_i f_N(n)=0.
\end{equation}
Note that for a fixed mean FPT, minimizing the variance of FPT and minimizing the second order moment $\left<FPT^2\right>$ are equivalent. 

\par Now, we consider the expression for $\left<FPT^2\right>$, and use expression in Eq. \eqref{eqn:epsiloncondition} to deduce the desired optimal function. From Eq. \eqref{eqn:varFPT}, we have: 
\begin{equation}
\left<FPT^2\right>=\sum_{n=1}^{\infty}\left(\sum_{i=1}^{n} \text{Var}\left(T_i\right) + \left(\sum_{i=1}^{n}\left<T_i\right>\right)^2 \right)f_N(n).
\end{equation}
Substituting value of $\displaystyle \left<T_i\right>$ from Eq. \eqref{eqn:assumedoptimalfunction}, we get following expression for $\displaystyle \left<FPT^2\right>$:
\small{
\begin{equation}
\left<FPT^2\right>=\sum_{n=1}^{\infty}\left(\sum_{i=1}^{n} \text{Var}\left(T_i\right) + \left(\sum_{i=1}^{n}\left(\frac{\left< FPT \right>}{\left<N\right>}+\epsilon_i\right)\right)^2\right)f_N(n).
\end{equation}
}
\normalsize
Further simplifying and using relation obtained in Eq. \eqref{eqn:epsiloncondition} yields: 
\small{
\begin{equation}
\label{eqn:meanFPTsqrd_I}
\left<FPT^2\right>=\frac{\left<FPT\right>^2}{\left<N\right>^2}\left<N^2\right>+\sum_{n=1}^{\infty}\left(\sum_{i=1}^{n} \text{Var}\left(T_i\right) + \left(\sum_{i=1}^{n}\epsilon_i\right)^2\right)f_N(n).
\end{equation}
}
\normalsize
Using Eq. \eqref{eqn:varmeanrelnTi} in Eq. \eqref{eqn:meanFPTsqrd_I}:
\begin{subequations}
\begin{align}
\left<FPT^2\right> &\geq \frac{\left<FPT\right>^2}{\left<N\right>^2}\left<N^2\right>\nonumber \\  &+\sum_{n=1}^{\infty}\left(\sum_{i=1}^{n} \left<T_i\right>^2 + \left(\sum_{i=1}^{n}\epsilon_i \right)^2\right)f_N(n), \label{eqn:meanFPTsqrd_II}\\
\implies \left<FPT^2\right> &\geq \frac{\left<FPT\right>^2}{\left<N\right>^2}\left<N^2\right>+\frac{\left<FPT\right>^2}{\left<N\right>} \nonumber \\  &+ \sum_{n=1}^{\infty}\left(\sum_{i=1}^{n} \epsilon_i^2 + \left(\sum_{i=1}^{n}\epsilon_i \right)^2\right)f_N(n). \label{eqn:meanFPTsqrd_III}
\end{align}
\normalsize
\end{subequations}
Further, we note that in above expression if $\epsilon_i=0$ (or equivalently $\left<\delta_i\right>=0$), the expression minimizes and reduces to: 
\begin{align}
\label{eqn:FPTsqFPTIneq}
\left<FPT^2\right> \geq \frac{\left<FPT\right>^2}{\left<N\right>^2}\left<N^2\right>+\frac{\left<FPT\right>^2}{\left<N\right>}.
\end{align}
Recalling Eq. \eqref{eqn:RelnFPTMomentsNoFB}, we observe that equality in above expression holds for unregulated gene expression case, which essentially means $\delta_i=0$. This proves the desired result. 
\end{proof}

\par In this section, we proved that having no auto-regulation of transcription rate provides minimum stochasticity in the FPT, if mean FPT and event threshold are kept fixed. However, since our analysis simplified the gene expression model to burst--limit, we are interested in validating whether it is true if we don't make an approximation. In the next section, we discuss the computer simulations we carried out for this purpose.
\section{Simulation Results}
\begin{figure*}[bt!]
\centering
{\includegraphics[width=\textwidth]{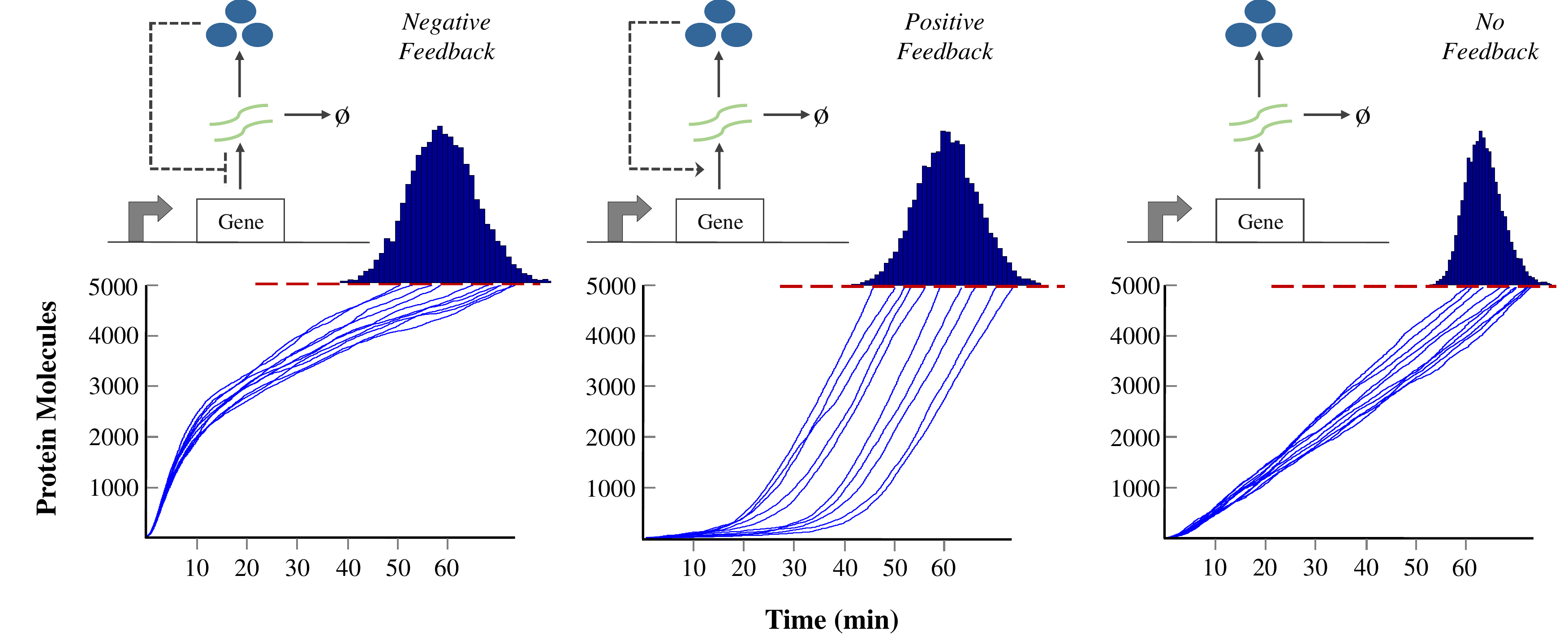}}
\caption{\textit{No protein--feedback regulation of transcription rate results in minimum stochasticity in FPT for a given mean and threshold}. In each figure, the dashed line in red represents the FPT threshold (assumed to be 5000 protein molecules here); the trajectories in the lower part depict the time evolution of protein population (10 sample trajectories); the histogram of on top represents distribution of FPT (10000 simulations); the parameters have been chosen to keep the mean FPT $\approx$ 60 min.}
\label{fig:simulationsresults}
\end{figure*}
In order to verify the result deduced in previous section, we carried out Monte Carlo simulations using Gillespie's algorithm \cite{Gillespie_ess_1977}. We did not specifically assume that production of protein is in geometric bursts with parameter $b$. Instead, we assumed a non--zero half--life for mRNA thereby relaxing the burst approximation.  

To simulate, we considered three separate cases: no feedback, negative feedback and positive feedback. The positive feedback is implemented using Hill function as follows:
\begin{equation}
k_m(j)=k_{\max} \left(r + (1-r) \frac{(jc)^H}{1+(jc)^H} \right),
\end{equation}
where $k_{\max}$ is maximum transcription rate, $r$ represents minimum transcription rate as the fraction of $k_{\max}$, $H$ denotes the Hill coefficient while $c$ is  coefficient proportional to the binding affinity (when $j=1/c,\; k_m(j) =k_{max}/2$). 
\par Similarly, the negative feedback is implemented using following function: 
\begin{equation}
k_m(j)=k_{\max} \left(r + (1-r) \frac{1}{1+(jc)^H} \right).
\end{equation}

We carried out the simulations for several sets of parameters assuming a fixed event threshold. Rest of the model parameters were chosen to keep the mean FPT approximately equal. In all of them, we found that no--feedback case has minimum variance in FPT.

\par In Table \ref{tab:parameters}, we present one set of such parameters. We assumed the event threshold $X=5000$. Other parameters are chosen in a way that the mean FPT $\approx$ $60$ minutes. 

Simulation results for 10000 realizations are shown in Fig. \ref{fig:simulationsresults}. We note that the variance is minimum in no--feedback case, validating our theoretical claims for this set of parameter values.
\section{Discussion}
In this work, we studied stochasticity in event timing at a single cell level. We considered a standard gene expression model without protein degradation. Next, we formulated the FPT problem for this model and derived the formulas for statistical moments of FPT. Further, we introduced auto-regulation in the gene expression wherein the transcription rate is a function of protein count. We derived the formulas for moments of FPT in this case as well, and demonstrated that for a given mean of FPT, the variance in FPT is minimized when there is no auto-regulation of gene expression. The result was verified with simulations as well.

\par The result can be connected to the $\lambda$ phage lysis time. Due to existence of optimal lysis time \cite{INWang_fitness_2006, Wang_evolutiontiming_1996}, the phage would possibly like to kill the cell at that time with as much precision as possible. Thus, it should resort to a strategy that would minimize the lysis time variance and hence have no protein--dependent feedback regulation of transcription rate in the expression of holin. In expression from late promoter in $\lambda$ phage, which produces holin, has no evidence of a regulation \cite{Ptashne_GeneticSwitch_1991, Oppenheim_switch_2005}.

\par Recalling that in no auto--regulation case too, the variance of FPT can be independently decreased by lowering the mean burst size $b$. Other studies also reveal that in case of $\lambda$ phage, the burst size is indeed small \cite{JohnDennehy_lst_2011, INWang_fitness_2006}. Also, antiholin, another protein expressed from the same promoter that expresses holin, binds to holin to decrease the effective burst size \cite{David_holinantiholin_2000,Abhi_ltv_2014}.
 
\par In this paper, there is an underlying assumption of protein being stable. In future work, we plan to use a gene expression model with protein degradation, and carry out a similar analysis. This can be further extended to more generalized gene expression models wherein the promoter can also switch between \textit{on} and \textit{off} states \cite{ShahrezaeVahid_ad_2008, abhi_transcriptionalburstingHIV_2010}.

\section*{Acknowledgment}
AS is supported by the National Science Foundation Grant DMS-1312926, University of Delaware Research Foundation (UDRF) and Oak Ridge Associated Universities (ORAU).

\bibliography{references}

\begin{thebibliography}{10}
\providecommand{\url}[1]{#1}
\csname url@rmstyle\endcsname
\providecommand{\newblock}{\relax}
\providecommand{\bibinfo}[2]{#2}
\providecommand\BIBentrySTDinterwordspacing{\spaceskip=0pt\relax}
\providecommand\BIBentryALTinterwordstretchfactor{4}
\providecommand\BIBentryALTinterwordspacing{\spaceskip=\fontdimen2\font plus
\BIBentryALTinterwordstretchfactor\fontdimen3\font minus
  \fontdimen4\font\relax}
\providecommand\BIBforeignlanguage[2]{{%
\expandafter\ifx\csname l@#1\endcsname\relax
\typeout{** WARNING: IEEEtran.bst: No hyphenation pattern has been}%
\typeout{** loaded for the language `#1'. Using the pattern for}%
\typeout{** the default language instead.}%
\else
\language=\csname l@#1\endcsname
\fi
#2}}

\bibitem{Blake_noise_2003}
W.~J. Blake, M.~Kaern, C.~R. Cantor, and J.~J. Collins, ``Noise in eukaryotic
  gene expression,'' \emph{Nature}, vol. 422, pp. 633--637, 2003.

\bibitem{Raser_noiseocc_2005}
J.~M. Raser and E.~K. O'Shea, ``Noise in gene expression: origins,
  consequences, and control,'' \emph{Science}, vol. 309, pp. 2010--2013, 2005.

\bibitem{ArjunRaj_nnc_2008}
A.~Raj and A.~van Oudenaarden, ``Nature, nurture, or chance: Stochastic gene
  expression and its consequences,'' \emph{Cell}, vol. 135, pp. 216 -- 226,
  2008.

\bibitem{Munsky_noisetoregulation_2012}
B.~Munsky, G.~Neuert, and A.~van Oudenaarden, ``Using gene expression noise to
  understand gene regulation,'' \emph{Science}, vol. 336, pp. 183--187, 2012.

\bibitem{Kaern_TP_2005}
M.~Kaern, T.~C. Elston, W.~J. Blake, and J.~J. Collins, ``{Stochasticity in
  gene expression: from theories to phenotypes},'' \emph{Nature Review
  Genetics}, vol.~6, pp. 451--64, 2005.

\bibitem{AbhiMohammad_var_2013}
A.~Singh and M.~Soltani, ``Quantifying intrinsic and extrinsic variability in
  stochastic gene expression models,'' \emph{PLoS ONE}, vol.~8, p. e84301, 12
  2013.

\bibitem{Losick_cellfate_2008}
R.~Losick and C.~Desplan, ``Stochasticity and cell fate,'' \emph{Science}, vol.
  320, pp. 65--68, 2008.

\bibitem{Arkin_ska_1998}
A.~Arkin, J.~Ross, and H.~McAdams, ``Stochastic kinetic analysis of
  developmental pathway bifurcation in phage lambda--infected escherichia coli
  cells,'' \emph{Genetics}, vol. 149, pp. 1633--1648, 1998.

\bibitem{Weinberger_lentiviral_2005}
L.~S. Weinberger, J.~C. Burnett, J.~E. Toettcher, A.~P. Arkin, and D.~V.
  Schaffer, ``Stochastic gene expression in a lentiviral positive-feedback
  loop: Hiv-1 tat fluctuations drive phenotypic diversity,'' \emph{Cell}, vol.
  122, pp. 169--182, 2005.

\bibitem{Veening_bistability_2008}
J.-W. Veening, W.~K. Smits, and O.~P. Kuipers, ``Bistability, epigenetics, and
  bet-hedging in bacteria,'' \emph{Annu. Rev. Microbiol.}, vol.~62, pp.
  193--210, 2008.

\bibitem{Hasty_switchamp_2000}
J.~Hasty, J.~Pradines, M.~Dolnik, and J.~J. Collins, ``Noise-based switches and
  amplifiers for gene expression,'' \emph{Proceedings of the National Academy
  of Sciences}, vol.~97, pp. 2075--2080, 2000.

\bibitem{abhi_transcriptionalburstingHIV_2010}
A.~Singh, B.~Razooky, C.~D. Cox, M.~L. Simpson, and L.~S. Weinberger,
  ``Transcriptional bursting from the hiv-1 promoter is a significant source of
  stochastic noise in hiv-1 gene expression,'' \emph{Biophysical Journal},
  vol.~98, pp. L32--L34, 2010.

\bibitem{Eldar_functional_2010}
A.~Eldar and M.~B. Elowitz, ``{Functional roles for noise in genetic
  circuits},'' \emph{Nature}, vol. 467, pp. 167--173, Sept. 2010.

\bibitem{Kussell_pd_2005}
E.~Kussell and S.~Leibler, ``{Phenotypic diversity, population growth, and
  information in fluctuating environments},'' \emph{Science}, vol. 309, pp.
  2075--2078, 2005.

\bibitem{Balaban_bactpersist_2004}
N.~Balaban, J.~Merrin, R.~Chait, L.~Kowalik, and S.~Leibler, ``{Bacterial
  persistence as a phenotypic switch},'' \emph{Science}, vol. 305, pp.
  1622--1625, 2004.

\bibitem{Murat_survival_2008}
M.~Acar, J.~T. Mettetal, and A.~van Oudenaarden, ``{Stochastic switching as a
  survival strategy in fluctuating environments},'' \emph{Nature Genetics},
  vol.~40, pp. 471--475, 2008.

\bibitem{Kemkemer_incnoise_2002}
R.~Kemkemer, S.~Schrank, W.~Vogel, H.~Gruler, and D.~Kaufmann, ``Increased
  noise as an effect of haploinsufficiency of the tumor-suppressor gene
  neurofibromatosis type 1 in vitro,'' \emph{Proceedings of the National
  Academy of Sciences}, vol.~99, pp. 13\,783--13\,788, 2002.

\bibitem{Cook_modeling_1998}
D.~L. Cook, A.~N. Gerber, and S.~J. Tapscott, ``Modeling stochastic gene
  expression: implications for haploinsufficiency,'' \emph{Proceedings of the
  National Academy of Sciences}, vol.~95, pp. 15\,641--15\,646, 1998.

\bibitem{Bahar_incvariation_2006}
R.~Bahar, C.~H. Hartmann, K.~A. Rodriguez, A.~D. Denny, R.~A. Busuttil, M.~E.
  Doll{\'e}, R.~B. Calder, G.~B. Chisholm, B.~H. Pollock, C.~A. Klein,
  \emph{et~al.}, ``Increased cell-to-cell variation in gene expression in
  ageing mouse heart,'' \emph{Nature}, vol. 441, pp. 1011--1014, 2006.

\bibitem{Lehner_noiseminm_2008}
B.~Lehner, ``Selection to minimise noise in living systems and its implications
  for the evolution of gene expression,'' \emph{Molecular systems biology},
  vol.~4, 2008.

\bibitem{Fraser_noiseminm_2004}
H.~B. Fraser, A.~E. Hirsh, G.~Giaever, J.~Kumm, and M.~B. Eisen, ``Noise
  minimization in eukaryotic gene expression,'' \emph{PLoS biology}, vol.~2, p.
  e137, 2004.

\bibitem{UriAlon_nm_2007}
U.~Alon, ``Network motifs: theory and experimental approaches,'' \emph{Nature
  Reviews Genetics}, vol.~8, pp. 450--461, 2007.

\bibitem{Becskei_engineeringstability_2000}
A.~Becskei and L.~Serrano, ``Engineering stability in gene networks by
  autoregulation,'' \emph{Nature}, vol. 405, pp. 590--593, 2000.

\bibitem{ElSamad_regulated_2006}
H.~El-Samad and M.~Khammash, ``Regulated degradation is a mechanism for
  suppressing stochastic fluctuations in gene regulatory networks,''
  \emph{Biophysical journal}, vol.~90, pp. 3749--3761, 2006.

\bibitem{Swain_attenuatestochasticity_2004}
P.~S. Swain, ``Efficient attenuation of stochasticity in gene expression
  through post-transcriptional control,'' \emph{Journal of Molecular Biology},
  vol. 344, pp. 965 -- 976, 2004.

\bibitem{Orrell_control_2004}
D.~Orrell and H.~Bolouri, ``Control of internal and external noise in genetic
  regulatory networks,'' \emph{Journal of theoretical biology}, vol. 230, pp.
  301--312, 2004.

\bibitem{abhi_fbstrength_2009}
A.~Singh and J.~P. Hespanha, ``Optimal feedback strength for noise suppression
  in autoregulatory gene networks,'' \emph{Biophysical journal}, vol.~96, pp.
  4013 -- 4023, 2009.

\bibitem{Tao_effectoffb_2007}
Y.~Tao, X.~Zheng, and Y.~Sun, ``Effect of feedback regulation on stochastic
  gene expression,'' \emph{Journal of Theoretical Biology}, vol. 247, pp. 827
  -- 836, 2007.

\bibitem{abhi_mRNA_2011}
A.~Singh, ``Negative feedback through mrna provides the best control of
  gene-expression noise,'' \emph{NanoBioscience, IEEE Transactions on},
  vol.~10, pp. 194--200, 2011.

\bibitem{Amir_noisetiming_2007}
A.~Amir, O.~Kobiler, A.~Rokney, A.~B. Oppenheim, and J.~Stavans, ``Noise in
  timing and precision of gene activities in a genetic cascade,''
  \emph{Molecular Systems Biology}, vol.~3, 2007.

\bibitem{Murugan_fluctuation_2011}
R.~Murugan and G.~Kreiman, ``On the minimization of fluctuations in the
  response times of autoregulatory gene networks,'' \emph{Biophysical Journal},
  vol. 101, pp. 1297--1306, 2011.

\bibitem{White_holintriggering_2011}
R.~White, S.~Chiba, T.~Pang, J.~S. Dewey, C.~G. Savva, A.~Holzenburg,
  K.~Pogliano, and R.~Young, ``Holin triggering in real time,''
  \emph{Proceedings of the National Academy of Sciences}, vol. 108, pp.
  798--803, 2011.

\bibitem{JohnDennehy_lst_2011}
J.~Dennehy and I.-N. Wang, ``Factors influencing lysis time stochasticity in
  bacteriophage lambda,'' \emph{BMC Microbiology}, vol.~11, no.~1, p. 174,
  2011.

\bibitem{Abhi_ltv_2014}
A.~Singh and J.~Dennehy, ``Stochastic holin expression can account for lysis
  time variation in the bacteriophage $\lambda$,'' \emph{Journal of the Royal
  Society Interface (to appear)}, 2014.

\bibitem{INWang_fitness_2006}
I.-N. Wang, ``Lysis timing and bacteriophage fitness,'' \emph{BMC
  Microbiology}, vol. 172, pp. 17--26, January 2006.

\bibitem{Wang_evolutiontiming_1996}
I.-N. Wang, D.~E. Dykhuizen, and L.~B. Slobodkin, ``The evolution of phage
  lysis timing,'' \emph{Evolutionary Ecology}, vol.~10, pp. 545--558, 1996.

\bibitem{Heineman_optimal_2007}
R.~Heineman and J.~Bull, ``Testing optimality with experimental evolution:
  lysis time in a bacteriophage,'' \emph{Evolution}, vol.~61, pp. 169;5--1709,
  2007.

\bibitem{Shao_adsorptionoptimal_2008}
Y.~Shao and I.-N. Wang, ``Bacteriophage adsorption rate and optimal lysis
  time,'' \emph{Genetics}, vol. 180, pp. 471--482, 2008.

\bibitem{Bonachela_optimallysis_2014}
J.~A. Bonachela and S.~A. Levin, ``Evolutionary comparison between viral lysis
  rate and latent period,'' \emph{Journal of Theoretical Biology}, vol. 345,
  pp. 32 -- 42, 2014.

\bibitem{redner2001guide}
S.~Redner, \emph{A guide to first-passage processes}.\hskip 1em plus 0.5em
  minus 0.4em\relax Cambridge University Press, 2001.

\bibitem{Shao_holinstable_2009}
Y.~Shao and N.~Wang, ``Effect of late promoter activity on bacteriophage
  $\lambda$ fitness,'' \emph{Genetics}, vol. 181, pp. 1467--1475, 2009.

\bibitem{Friedman_pd_2006}
N.~Friedman, L.~Cai, and X.~S. Xie, ``Linking stochastic dynamics to population
  distribution: An analytical framework of gene expression,'' \emph{Phys. Rev.
  Lett.}, vol.~97, p. 168302, 2006.

\bibitem{ShahrezaeVahid_ad_2008}
V.~Shahrezaei and P.~S. Swain, ``Analytical distributions for stochastic gene
  expression,'' \emph{Proceedings of the National Academy of Sciences}, vol.
  105, pp. 17\,256--17\,261, 2008.

\bibitem{Paulsson_sge_2005}
J.~Paulsson, ``Models of stochastic gene expression,'' \emph{Physics of Life
  Reviews}, vol.~2, pp. 157 -- 175, 2005.

\bibitem{Berg_statfluct_1978}
O.~G. Berg, ``A model for the statistical fluctuations of protein numbers in a
  microbial population,'' \emph{Journal of Theoretical Biology}, vol.~71, pp.
  587 -- 603, 1978.

\bibitem{YuXiao_pom_2006}
J.~Yu, J.~Xiao, X.~Ren, K.~Lao, and X.~S. Xie, ``Probing gene expression in
  live cells, one protein molecule at a time,'' \emph{Science}, vol. 311, pp.
  1600--1603, 2006.

\bibitem{Elgart_connecting_2011}
V.~Elgart, T.~Jia, A.~T. Fenley, and R.~Kulkarni, ``Connecting protein and mrna
  burst distributions for stochastic models of gene expression,''
  \emph{Physical biology}, vol.~8, p. 046001, 2011.

\bibitem{Degroot_prob_2012}
M.~H. DeGroot and M.~J. Schervish, \emph{Probability and Statistics},
  4th~ed.\hskip 1em plus 0.5em minus 0.4em\relax Pearson, 2012.

\bibitem{Papoulis_prob_2002}
A.~Papoulis and S.~U. Pillai, \emph{Probability, Random Variables and
  Stochastic Processes}, 4th~ed.\hskip 1em plus 0.5em minus 0.4em\relax McGraw
  Hill, 2002.

\bibitem{Spiegel_scham_1992}
M.~R. Spiegel, \emph{Theory and Problems of Probability and Statistics}.\hskip
  1em plus 0.5em minus 0.4em\relax McGraw-Hill, 1992.

\bibitem{Ross_prob_2010}
S.~K. Ross, \emph{Introduction to Probability Models}, 10th~ed.\hskip 1em plus
  0.5em minus 0.4em\relax Academic Press, 2010.

\bibitem{Gillespie_ess_1977}
D.~Gillespie, ``Exact stochastic simulation of coupled chemical reactions,''
  \emph{J Phys Chem}, vol.~81, pp. 2340--2361, 1977.

\bibitem{Ptashne_GeneticSwitch_1991}
M.~Ptashne, \emph{A Genetic Switch -- Phage $\lambda$ and Higher Organisms},
  2nd~ed.\hskip 1em plus 0.5em minus 0.4em\relax Cell Press \& Blackwell
  Scientific Publications, 1991.

\bibitem{Oppenheim_switch_2005}
A.~B. Oppenheim, O.~Kobiler, J.~Stavans, D.~L. Court, and S.~Adhya, ``Switches
  in bacteriophage lambda development,'' \emph{Annual Review of Genetics},
  vol.~39, pp. 409--429, 2005.

\bibitem{David_holinantiholin_2000}
D.~L. Smith, U.~Blasi, and R.~Young, ``Dimerization between the holin and holin
  inhibitor of phage $\lambda$,'' \emph{Journal of Bacteriology}, vol. 182, pp.
  6075--6081, 2000.

\end{thebibliography}
\bibliographystyle{IEEEtran}
\end{document}